\documentclass[12pt]{article}
\usepackage{amsmath}
\usepackage{amssymb}
\usepackage{mathrsfs}
\usepackage{amsthm}
\usepackage{rotating} 
\usepackage{authblk}
\usepackage{booktabs, longtable, array}
\usepackage{color}
\usepackage{graphicx}
\usepackage{caption}
\usepackage{subcaption}
\usepackage{soul}
\usepackage{bbm} 
\usepackage{ulem}
\usepackage{cancel} 
\usepackage{algorithm}
\usepackage{algorithmicx}
\usepackage{amssymb}
\usepackage{threeparttable}
\usepackage{titlesec}
\usepackage{tikz}
\usepackage{bbding}
\usepackage{pifont}
\usepackage{titletoc}
\usepackage{hhline}

\theoremstyle{definition}
\allowdisplaybreaks[4]
\usepackage[margin=1.2in]{geometry}

\newtheorem{definition}{Definition}
\newtheorem{theorem}{Theorem}
\newtheorem{lemma}{Lemma}

\newtheorem{remark}{Remark}
\newtheorem{proposition}{Proposition}

\usepackage{hyperref}

\newcommand{\e}{\pmb{e}}

\newcommand{\x}{\pmb{x}}

\newcommand{\X}{\pmb{X}}

\newcommand{\Z}{\pmb{Z}}
\newcommand{\xu}{\pmb{u}}

\newcommand{\pp}{\pmb{\xi}}

\newcommand{\be}{\begin{equation}}
\newcommand{\ee}{\end{equation}}
\newcommand{\bea}{\begin{equation*}}
\newcommand{\eea}{\end{equation*}}

\newcommand{\E}{\mathbb{E}}

\newcommand{\R}{\mathbb{R}}

\newcommand{\trace}{\text{Tr}\,}

\newcommand{\mA}{\mathcal{A}}
\newcommand{\mT}{\mathcal{T}}
\newcommand{\mD}{\mathcal{D}}
\newcommand{\mK}{\mathcal{K}}
\newcommand{\mO}{\mathcal{O}}

\def\lV{\left\lVert}
\def\rV{\right\lVert}
\def\lv{\left\lvert}
\def\rv{\right\lvert}
\def\lk{\left(}
\def\rk{\right)}
\def\lg{\langle}
\def\rg{\rangle}
\def\lz{\left[}
\def\rz{\right]}

\begin{document}
	\title{ Low-Rank Toeplitz Matrix Restoration: Descent Cone Analysis and Structured Random Matrix\footnote{This work was supported in part by the NSFC under grant numbers U21A20426 and 12071426}}

\author{Gao Huang\footnote{ School of Mathematical Science, Zhejiang University, Hangzhou 310027, P. R. China, E-mail address: hgmath@zju.edu.cn}}	
\author{Song Li\footnote{ Corresponding Author, School of Mathematical Science, Zhejiang University, Hangzhou 310027, P. R. China, E-mail address: songli@zju.edu.cn}}	
\date{}
\affil{}
\renewcommand*{\Affilfont}{\small\it}
	
	\maketitle
	\begin{abstract}
	This note demonstrates that we can stably recover all symmetric Toeplitz matrices $\X_0\in\R^{n\times n}$ of rank at most $r$ from a number of rank-one subgaussian measurements on the order of $r\log^{2} n$ with an exponentially decreasing failure probability by employing a nuclear norm minimization program. 
	Our approach utilizes descent cone analysis through Mendelson's small ball method with the Toeplitz constraint.
	The key ingredient is to determine the spectral norm of a random matrix with Toeplitz structure, which may be of independent interest. 
	This improves upon earlier analyses and resolves the conjecture in Chen et al. (IEEE Transactions on Information Theory, 61(7):4034--4059, 2015).

	\end{abstract}		
	
\section{Introduction}

Toeplitz matrices enjoy an important number of applications, including signal processing, numerical analysis, stochastic processes, time series analysis, image processing, and so on \cite{chan1996conjugate}. 
Recovery of a matrix that is simultaneously low-rank and Toeplitz, which arises when a random process is generated by a few spectral spikes,
is crucial for many tasks in wireless communications and array signal processing; see, e.g., \cite{6867345,chen2015exact}. 

Following some ideas in \cite{chen2015exact}, we aim to reconstruct an unknown symmetric matrix $\X_{0}\in\R^{n\times n}$ that is simultaneously low-rank and Toeplitz from a small number of rank-one measurements. 
In particular, we explore sampling methods of the form
\begin{equation*}\label{meas}
		b_{k}=\lg \pp_k\pp_{k}^{\mathsf{T}},\X_{0}\rg+e_{k},\quad k=1,\ldots,m,
		\end{equation*}
where $\pmb{b}:=\left\{b_{k}\right\}_{k=1}^{m}$ denotes the measurements, $\pp_{k}\in\R^{n}$ represents the sensing vector, $\e:=\left\{e_{k}\right\}_{k=1}^{m}$ stands for the noise term, and $m$ is the number of measurements. 
The rank-one sampling scheme is utilized in a wide spectrum of practical settings, 
such as covariance sketching for data streams, noncoherent energy measurements in communications, and phaseless measurements in optical imaging \cite{chen2015exact,cai2015rop}.
Furthermore, it admits near-optimal estimation with tractable algorithms and offers computational and storage advantages over other types of measurements.

Let us formally describe our setup and model, in accordance with \cite{chen2015exact}. 
We focus exclusively on the random sampling model. 
Specifically, we assume that the sensing vectors are independent copies of a random vector $\pp$, whose entries are i.i.d. copies of a random variable $\xi$ with
subgaussian norm $K$\footnote{The subgaussian norm of a random variable $X$ is defined as $\lV X\rV_{\psi_2} := \inf\{ t >0: \E \exp\lk\lv X\rv^2/t^2\rk \le 2\}. $} and satisfy the moment properties:
\begin{equation}\label{mom}
\E\xi=0,\ \E\xi^2=1, \ \text{and}\ \mu:=\E\xi^4\ge1.
	\end{equation}	
We assume that the noise vector $\e$ is bounded in $\ell_p$-norm ($p\ge1$). 
We also define the linear operator $\mA:\R^{n\times n}\rightarrow\R^{m}$ that maps a matrix $\X\in \R^{n\times n}$ to $\left\{\lg \pp_{k}\pp_{k}^{\mathsf{T}},\X\rg\right\}_{k=1}^{m}$. 
Thus, we express the measurements as
 \begin{equation*}
 \pmb{b}=\mA\lk\X_0\rk+\e.
  \end{equation*}
  
A natural heuristic method is rank minimization to promote low-rank structures. 
However, the rank minimization problem is generally considered to be NP-hard.
An alternative scheme is to seek a nuclear norm minimizer over all matrices compatible with the measurements, as well as the Toeplitz constraint:
 \begin{equation} \label{model}
		\begin{array}{ll}
			\text{minimize}  & \quad \lV\X\rV_{*}\\
			\text{subject to} & \quad\lV\mA\lk\X\rk-\pmb{b}\rV_{\ell_p}\le\eta\\	
			                         & \quad \X \ \text{is Toeplitz}.\\			 
		\end{array}
	\end{equation}
Here $\eta> 0$ is an a priori bound for the noise level; that is, we assume $\lV\e\rV_{\ell_p}\le\eta$.

Before we elaborate on our theorem, let's first introduce the previous work \cite{chen2015exact}.
To address the program (\ref{model}), which involves an excess of positive semidefinite cone constraints, Chen et al. in  \cite{chen2015exact} employed the stochastic RIP-$\ell_2/\ell_2$ condition.
They demonstrated that when the noise is bounded in the $\ell_2$-norm, a rank-$r$ symmetric Toeplitz matrix $\X_0\in \R^{n\times n}$
can be stably recovered from a number of rank-one subgaussian measurements on the order of $r\log^{10} n$.
This result was inspired by the Gaussian process argument using chaining techniques described in \cite{rudelson2008sparse}.
Their method can be applied to a broader range of settings—for instance, the recovery of low-rank matrices with random subsampling using a Fourier-type basis, as well as the recovery of matrices that are both low-rank and Toeplitz under i.i.d. Gaussian measurement systems.

Nevertheless, as pointed out in \cite{chen2015exact}, their method has certain limitations when applied to Toeplitz matrices, leading to some unresolved conjectures. 
First, their method has limited applicability for rank-one measurements. 
As $\mu=\E \xi^4\ge\lk\E \xi^2\rk^2=1$, they are only able to provide theoretical guarantees for $1<\mu\le3$;
in other words, the tails of the measurement distributions cannot be heavier than those of the Gaussian distribution (e.g., $\mu=3$ for the Gaussian distribution), and the Bernoulli distribution is also excluded (e.g., $\mu=1$ for the Bernoulli distribution).
The main challenge lies in the fact that rank-one measurements do not satisfy the stochastic RIP-$\ell_2/\ell_2$ condition (see, e.g., \cite{candes2013phaselift}) and require the construction of equivalent sampling operators to meet the assumptions in \cite[Theorem 5]{chen2015exact}.
Secondly, unlike other low-rank matrix recovery problems \cite{chandrasekaran2012convex,kueng2017low}, their method does not achieve recovery with an exponentially decreasing failure probability but instead with a polynomially decreasing one, which is also due to the stochastic RIP-$\ell_2/\ell_2$ condition. 
For these two issues, they conjectured in \cite{chen2015exact} that the result can be improved through other methods
\footnote{In \cite[Remark 2]{chen2015exact}, it is stated that “We conjecture that these two aspects can be improved via other proof techniques.”}.\\

\noindent\textbf{Conjecture.}
Suppose the sample matrices are $\left \{\pp_{k}\pp^{\mathsf{T}}_{k} \right\}_{k=1}^{m}$, where $\pp_k$ are independent copies of a random vector $\pp$, whose entries are i.i.d. copies of a subgaussian random variable $\xi$ satisfying the moment condition (\ref{mom}). 
Then for all $\mu=\E \xi^4\ge1$ and with probability exceeding $1-e^{-cm}$, we can recover low-rank Toeplitz matrices through program (\ref{model}) with a near-optimal sample size, where $c>0$ is a numerical absolute constant.\\

To resolve the aforementioned conjecture, we adopt descent cone analysis via Mendelson’s small ball method, but confine it to the Toeplitz matrix space,
abandoning the stochastic RIP-$\ell_2/\ell_2$ approach.
Over the past decade, a number of works have studied sparse recovery \cite{lecue2017sparse,mendelson2018improved}, low-rank matrix recovery \cite{chandrasekaran2012convex,kueng2017low}, phase retrieval \cite{eldar2014phase,tropp2015convex,krahmer2020complex}, and blind deconvolution \cite{krahmer2021convex} via descent cone analysis.
Descent cone analysis via Gordon’s “escape through a mesh” theorem is suitable for i.i.d. Gaussian measurements \cite{chandrasekaran2012convex}, 
but it suffers from pitfalls when dealing with rank-one measurements.
In contrast, Mendelson’s small ball method \cite{koltchinskii2015bounding} is a powerful strategy for establishing a lower bound on a nonnegative empirical process that can address rank-one subgaussian measurements \cite{kueng2017low,krahmer2020complex}.

Two alternative approaches for rank-one measurements are the rank RIP-$\ell_1/\ell_2$ condition \cite{chen2015exact,cai2015rop} and dual certificate analysis \cite{candes2013phaselift}. 
However, the bottleneck of the former lies in establishing a precise upper bound on the covering number of all Toeplitz matrices of rank $r$, see e.g., \cite{chen2015exact};
while the latter resides in constructing dual elements, and the number of measurements is typically not optimal with respect to $r$, see e.g., \cite{6867345}.
The small ball method confined to the Toeplitz matrix space will achieve a more effective solution than either of these strategies.
The following theorem provides a resolution to the conjecture.

\begin{theorem}\label{thm}
Let $p\ge1$ and $\lV\e\rV_{\ell_p}\le \eta$. 
Suppose the sample matrices $\left \{\pp_{k}\pp^{\mathsf{T}}_{k} \right\}_{k=1}^{m}$ are such that
$\pp_k$ are independent copies of a random vector $\pp$, whose entries are i.i.d. copies of a random variable $\xi$ with subgaussian norm $K$ and  satisfy the moment condition (\ref{mom}). 
For all $\mu\ge1$, with probability exceeding $1-e^{-c m}$,  the solution $\widehat{\X}$ to (\ref{model}) satisfies
   \begin{equation}\label{model sulution}
	        \lV\widehat{\X}-\X_{0}\rV_{F} \le C \frac{\eta}{m^{1/p}},
	         \end{equation}	
		simultaneously for all symmetric Toeplitz matrices $\X_0$ of rank at most $r$, 
		provided that $m\ge Lr\log^{2}n$.
	Here, $c,C$ and $L$ are constants depending only on $K$ and $\mu$.	
	\end{theorem}
\begin{remark}
\cite[Theorem 2]{chen2015exact} is applicable to the model that includes positive semidefinite cone constraints; as mentioned in \cite{chen2015exact}, our result is also suitable for this scenario.
\end{remark}	

Our theorem improves the analyses in \cite{chen2015exact} and provides a positive answer to the aforementioned conjecture.
Firstly, our result applies to all subgaussian measurements where $\mu\ge1$, not just $1<\mu\le3$.	
Note that Bernoulli measurements are also applicable, which is different from the fact that in phase retrieval (when $\X_0=\x_0\x_{0}^{\mathsf{T}}$ is rank-1),
it is impossible to distinguish between vectors $\pmb{e}_{1}$ and $\pmb{e}_{2}$ from Bernoulli measurements; see, e.g., \cite{krahmer2017phase,krahmer2020complex}.
This is due to the fact that $\pmb{e}_{1}\pmb{e}^{\mathsf{T}}_{1}$ and $\pmb{e}_{2}\pmb{e}^{\mathsf{T}}_{2}$ correspond to the same Toeplitz matrix such that
$\mT\lk\pmb{e}_{1}\pmb{e}^{\mathsf{T}}_{1}\rk=\mT\lk\pmb{e}_{2}\pmb{e}^{\mathsf{T}}_{2}\rk=\frac{1}{n}\pmb{I}$,
where $\mT$ is the orthogonal projection operator onto Toeplitz matrices.
Aside from that, the result stabilizes with an exponentially decreasing failure probability. 
Since $m\ge L r \log^2 n$, the failure probability satisfies  $e^{-c m}\le\exp\lk-cLr\log^2n\rk$, which is super-polynomially small in $n$ for fixed $r$. 
Thus, an exponentially decreasing failure probability is significant.
Besides, the sampling complexity in \cite[Theorem 2]{chen2015exact} requires $m=\mO\lk r\log^{10} n\rk$, which is optimal for the rank constraint 
$r$ but remains far from the optimal sampling order.
Our result improves the number of measurements from $m=\mO\lk r\log^{10} n\rk$ to $m=\mO\lk r\log^{2} n\rk$.
Furthermore, our result has the capability to address a broader range of $\ell_p$-norm bounded noise.

Our proof consists of three ingredients. 
The first ingredient is that the descent cone for the nuclear norm has a clear characterization.
The second is that the small ball function in the small ball method, when confined to Toeplitz matrices, is bounded away from zero for rank-one subgaussian measurements.
Finally, we provide an upper bound for the spectral norm of a random matrix with Toeplitz structure.
Our approach relies on embedding Toeplitz matrices into circulant matrices, allowing for the explicit expression of the latter's eigenvalues, 
which diverges from the moment method \cite{bandeira2023matrix} and the random process method \cite{vershynin2018high}, offering an alternative perspective that may be of independent interest.

 \section{Proofs}
We prove our results in this section.
Before we begin, we introduce some notation. 
The symmetric Toeplitz matrix considered in this paper refers to the following form
\begin{equation*}
 \left(\begin{array}{cccc}
 z_0&z_1&\cdots & z_{n-1}\\
 z_{1}& \ddots&\ddots&\vdots\\
\vdots &  \ddots &\ddots&z_1\\
 z_{n-1}&\cdots&z_1&z_{0}\\ 
\end{array}\right).
\end{equation*}

We recall that we assume the sample vectors $\{\pp_{k}\}_{k=1}^{m}$ in sample matrices $\{\pp_{k}\pp_{k}^{\mathsf{T}}\}_{k=1}^{m}$ are independent copies of a random vector $\pp$. 
Specifically, the entries of $\pp$ are i.i.d. copies of a random variable $\xi$ with the following properties:
\begin{equation*}
\E\xi=0,\E\xi^2=1, \E\xi^4=\mu\ge1\ \text{and}\ \lV\xi\rV_{\psi_{2}}\le K.
	\end{equation*}

For a vector $\pmb{x}$, $\lV\x\rV_{\ell_p}$ denotes the usual $\ell_p$-norm. For a rank-$r$ matrix $\X$, we use $\lV\X\rV_{*}$ to denote its nuclear norm, $\lV\X\rV_{F}$ to denote its Frobenius norm and $\lV\X\rV_{op}$
to denote its spectral norm.
We denote its corresponding singular value decomposition (SVD) by $\X=\pmb{U}\pmb{\Sigma}\pmb{V}^{*}$,
where $\pmb{\Sigma}\in\R^{r \times r}$ is a diagonal matrix with nonnegative entries and $\pmb{U},\pmb{V}\in\R^{n\times r}$ satisfy 
$\pmb{U}^{*}\pmb{U}=\pmb{V}^{*}\pmb{V}=\pmb{I}_{r}$.
Then the tangent space of the manifold of rank-$r$ matrices at $\X$ is defined by
\begin{equation*}
T_{\X}:=\left\{\pmb{U}\pmb{A}^{*}+\pmb{B}\pmb{V}^{*}:\pmb{A},\pmb{B}\in\R^{n\times r}\right\}.
	\end{equation*}
We denote the orthogonal complement to $T_{\X}$ as $T_{\X}^{\perp}$.
In the following, we use $\Z_{T_{\X}}$ to denote the orthogonal projection of $\Z$ onto $T_{\X}$ and $\Z_{T^{\perp}_{\X}}$ to denote the orthogonal projection of $\Z$ onto $T^{\perp}_{\X}$.
In addition, $\mathbb{S}_{F}$ denotes the Frobenius unit sphere of $\R^{n\times n}$,
$\mathcal{S}^{n}$ denotes the vector space of all symmetric matrices in $\R^{n\times n}$,
$\mT$ denotes the orthogonal projection operator onto Toeplitz matrices.

\subsection{Preliminaries}
We present the Hanson-Wright inequality, which gives a concentration bound for quadratic forms of random variables;
for instance, see \cite[Theorem 6.2.1]{vershynin2018high}.
     \begin{lemma}\label{hanson}
      Let $\pp=\lk\xi_{1},\ldots,\xi_{n}\rk^{\mathsf{T}}$ 
      be a random vector with independent components satisfying $\E\xi_i=0$, $\E\xi_i^2=1$, and $\max_i\lV\xi_i\rV_{\psi_{2}}\le K$. 
      Then for any $\Z\in\mathcal{S}^{n}$, there is a numerical constant $c>0$ such that for all $t>0$ 
		\begin{equation}
	\mathbb{P} \lk \lv  \pmb{\xi}^{\mathsf{T}}\Z\pmb{\xi} - \mathbb{E}\lz \pmb{\xi}^{\mathsf{T}}\Z\pmb{\xi} \rz\rv > t   \rk 
	\le 2  \exp \lk-c\min \left\{ \frac{t^{2}}{K^{4} \lV\Z\rV^{2}_{F}}  , \frac{t}{K^{2} \lV \Z \rV_{op}}  \right\} \rk.	
	\end{equation}
	\end{lemma}
	
	We next present Mendelson's small ball method, which can provide a lower bound for nonnegative empirical processes
that can address a broad class of complexity measures and sampling matrices, especially for heavy-tailed measurements.
	The proof can be found in \cite[Theorem 2.1]{koltchinskii2015bounding} or \cite[Proposition 5.1]{tropp2015convex}.
	\begin{lemma}\label{small ball}
	Fix $\mK\subset\R^{n}$ and
	let $\{\pmb{\phi}_k\}_{k=1}^{m}$ be independent copies of a random vector $\pmb{\phi}$ in $\R^{n}$. 
	Consider the small ball function
    \begin{equation}	     
     \mathcal{Q}_{\alpha}\lk \mK;\pmb{\phi}\rk=\inf_{\xu\in\mK}\mathbb{P}\lk\lv \lg\pmb{\phi},\xu\rg\rv\ge\alpha\rk 
	         \end{equation}
	         and the supremum of the empirical process
	          \begin{equation}	    
	            \mathcal{R}_{m}\lk \mK;\pmb{\phi}\rk
	            =\E\sup_{\xu\in\mK}\lv\frac{1}{m}\sum_{k=1}^{m}\varepsilon_{k}\lg\pmb{\phi}_k,\xu\rg\rv,
	      	         \end{equation} 
where $\{\varepsilon_{k}\}_{k=1}^{m}$ is a Rademacher sequence independent of everything else.        
		         
		         Then for any $p\ge1,\alpha>0$ and $t> 0$, with probability exceeding $1-\exp\lk-2t^{2}\rk$,
	  \begin{equation}
\inf_{\xu\in\mK}\lk\sum_{k=1}^{m}\lv\lg\pmb{\phi}_k,\xu\rg\rv^{p}\rk^{1/p} 
	\ge
	m^{1/p}\lk\alpha\mathcal{Q}_{2\alpha}\lk \mK;\pmb{\phi}\rk-2\mathcal{R}_{m}\lk \mK;\pmb{\phi}\rk-\frac{\alpha t}{\sqrt{m}}\rk.
	    \end{equation}     
\end{lemma}

 \subsection{Descent Cone Analysis}
 In this subsection, we present a framework for applying the descent cone analysis method and a detailed examination of the small ball method as it pertains to Toeplitz matrices.
 
The descent cone of a norm at a point $\X_0\in \R^{n\times n}$ is the set of all possible directions $\Z\in\R^{n\times n}$
such that the norm does not increase. 
For the nuclear norm, this leads to the following definition.
 \begin{definition}
 For any matrix $\X_0\in \R^{n\times n}$, define its descent cone $\mD\lk \X_0\rk$ by
     \begin{equation}	     
     \mD\lk \X_0\rk:=\left\{\Z\in\R^{n\times n}:\lV\X_0+t\Z\rV_{*}\le\lV\X_0\rV_{*} \text{for some}\ t>0\right\}.
     	         \end{equation}
 \end{definition}
 
The first crucial ingredient for Theorem \ref{thm} is the following proposition that characterizes the geometry property of the descent cone. 
The proof is inspired by \cite[Lemma 4.1]{krahmer2021convex}, and it can serve as a substitute for \cite[Lemma 10]{kueng2017low}. 
\begin{proposition}\label{des}
Let $\X_0\in\mathcal{S}^{n}$ be of rank at most $r$. Then
\begin{equation}
\lV\Z\rV_{*}\le\lk\sqrt{2}+1\rk\sqrt{r} \lV\Z\rV_{F},\ \text{for all}\ \Z\in\overline{\mD\lk \X_0\rk},
	    \end{equation}  
	    where $\overline{\mD\lk \X_0\rk}$ denotes the topological closure of $\mD\lk \X_0\rk$.
\end{proposition}
\begin{proof}
We denote the SVD of $\X_0$ by $\X_0=\pmb{U}\pmb{\Sigma}\pmb{V}^{*}$.	
Then \cite[Lemma 4.1]{krahmer2021convex} states that
\begin{equation}\label{dual}	
            \overline{\mD\lk \X_0\rk}=\left\{ \Z\in\R^{n\times n}:-\lg\pmb{U}\pmb{V}^{*},\Z\rg\ge\lV\Z_{T_{\X_0}^{\perp}}\rV_{*}\right\}.
	\end{equation}
Thus, for any  $\Z\in\overline{\mD\lk \X_0\rk}$, we have
\begin{equation*}	
	\begin{aligned}
	    \lV\Z\rV_{*}&\le \lV\Z_{T_{\X_0}}\rV_{*}+ \lV\Z_{T_{\X_0}^{\perp}}\rV_{*}\\
	                &\le \sqrt{2r}\lV\Z_{T_{\X_0}}\rV_{F}-\lg\pmb{U}\pmb{V}^{*},\Z\rg\\
	                &\le \sqrt{2r}\lV\Z_{T_{\X_0}}\rV_{F}+\lV\pmb{U}\pmb{V}^{*}\rV_{F}\lV\Z\rV_{F}\\
	                &\le \sqrt{2r}\lV\Z\rV_{F}+\lV\pmb{U}\rV_{op}\lV\pmb{V}^{*}\rV_{F}\lV\Z\rV_{F}\\
	                &\le \lk\sqrt{2r}+\sqrt{r}\rk\lV\Z\rV_{F}. 
	                \end{aligned}
	\end{equation*}	
	In the second inequality, we have used the fact that $\Z_{T_{\X_0}}$ has rank at most $2r$ and (\ref{dual}). 
	The fourth inequality follows from the fact that, for any $\pmb{U},\pmb{V}$, we have $\lV\pmb{U}\pmb{V}^{*}\rV_{F}\le\lV\pmb{U}\rV_{op}\lV\pmb{V}^{*}\rV_{F}$.
\end{proof}

In the noiseless scenario, i.e., $\eta=0$, the matrix $\X_0$ is the unique minimizer of the semidefinite program (\ref{model})
if and only if the null space of $\mA$ does not intersect the descent cone $\mD\lk \X_0\rk$ except at the origin.
Then if the following quantity for a matrix $\X_0$ is bounded away from $0$,
 \begin{equation*}	     
    \lambda_{\min}\lk\mA;\mD\lk \X_0\rk\cap \mathbb{S}_{F}\rk:=\inf_{\Z\in\mD\lk \X_0\rk \backslash \left\{0\right\}\cap \mathbb{S}_{F}}\frac{\lV\mA\lk\Z\rk\rV_{\ell_p}}{\lV\Z\rV_{F}},
	         \end{equation*} 
	        which is often referred to as the minimum conic singular value \cite{chandrasekaran2012convex,tropp2015convex,kueng2017low,krahmer2021convex}, 
            the intersection of the null space of $\mA$ and $\mD\lk \X_0\rk$ is empty except at the origin.
For our specific situation, we only need to focus on the Toeplitz matrices that fall within the descent cone $\mD\lk \X_0\rk$. 
Therefore, we only need to evaluate the minimum conic singular value within the set  $\mD\lk \X_0\rk\cap \mT$. 
On the other hand, to ensure universal outcomes for all Toeplitz matrices of rank at most $r$, the descent cone $\mD\lk \X_0\rk$
can be extended to 
 \begin{equation}	     
    \mK^{r}_{\mT}:=\bigcup_{\X_0}\lk\mD\lk \X_0\rk\cap \mT\rk,
	         \end{equation}
	         where the union runs over all Toeplitz matrices $\X_0\in\mathcal{S}^{n} \backslash \left\{0\right\}$ of rank at most $r$.
	         Therefore, the minimum conic singular value that this paper focuses on is
	 \begin{equation}	     
    \lambda_{\min}\lk\mA; \mK^{r}_{\mT}\cap \mathbb{S}_{F}\rk:=\inf_{\Z\in \mK^{r}_{\mT} \backslash \left\{0\right\}\cap \mathbb{S}_{F}}\frac{\lV\mA\lk\Z\rk\rV_{\ell_p}}{\lV\Z\rV_{F}}.
	         \end{equation}         
	       
	       We provide a concise description to prove Theorem \ref{thm} through $\lambda_{\min}\lk\mA; \mK^{r}_{\mT}\cap \mathbb{S}_{F}\rk$.
It follows from the triangle inequality that
 \begin{equation*}	     
    \lV\mA\lk\widehat{\X}-\X_0\rk\rV_{\ell_p}\le \lV\mA\lk\widehat{\X}\rk-\pmb{b}\rV_{\ell_p}+\lV\mA\lk\X_{0}\rk-\pmb{b}\rV_{\ell_p}\le2\eta.
	         \end{equation*}
	         Thus, as presented in \cite[Proposition 2.2]{chandrasekaran2012convex}, we have that
	 \begin{equation}\label{aaa}	
	 \begin{aligned}     
   \lV\widehat{\X}-\X_0\rV_{F}&=\frac{\lV\mA\lk\widehat{\X}-\X_0\rk\rV_{\ell_p}}{\lV\mA\lk\widehat{\X}-\X_0\rk\rV_{\ell_p}/\lV\widehat{\X}-\X_0\rV_{F}}\\
   &\le\frac{2\eta}{  \lambda_{\min}\lk\mA; \mK^{r}_{\mT}\cap \mathbb{S}_{F}\rk}.
   \end{aligned}
	         \end{equation}         
Equation (\ref{aaa}) indicates that to prove Theorem \ref{thm}, we need to show that $\lambda_{\min}\lk\mA; \mK^{r}_{\mT}\cap \mathbb{S}_{F}\rk$ is bounded away from zero.
Mendelson’s small ball method can lead to the geometric analysis of $\lambda_{\min}\lk\mA; \mK^{r}_{\mT}\cap \mathbb{S}_{F}\rk$.
By Lemma \ref{small ball}, with probability exceeding $1-\exp\lk-2t^{2}\rk$ we have that
 \begin{equation}\label{small1}
 \begin{aligned}
\lambda_{\min}&\lk\mA; \mK^{r}_{\mT}\cap \mathbb{S}_{F}\rk\\
&	\ge
	m^{1/p}\lk\alpha \mathcal{Q}_{2\alpha}\lk \mK^{r}_{\mT}\cap \mathbb{S}_{F};\pmb{\xi}\pmb{\xi}^{\mathsf{T}}\rk-2\mathcal{R}_{m}\lk \mK^{r}_{\mT}\cap \mathbb{S}_{F};\pp\pp^{\mathsf{T}}\rk-\frac{\alpha t}{\sqrt{m}}\rk.
	\end{aligned}
	  \end{equation}  
	  Thus, we need a lower bound on the small ball function $\mathcal{Q}_{2\alpha}\lk \mK^{r}_{\mT}\cap \mathbb{S}_{F};\pmb{\xi}\pmb{\xi}^{\mathsf{T}}\rk$ and an upper bound on the supremum of the empirical process 
	  $\mathcal{R}_{m}\lk \mK^{r}_{\mT}\cap \mathbb{S}_{F};\pp\pp^{\mathsf{T}}\rk$.

The second ingredient is that when confined to the space of the Toeplitz matrices $\mT$ (it can be seen that $\mK^{r}_{\mT}\subset \mT$), the small ball function is bounded away from zero for all $\mu\ge1$.
\begin{theorem}\label{small}
Let $\pp=\lk\xi_{1},\ldots,\xi_{n}\rk^{\mathsf{T}}$ be a random vector whose entries are i.i.d. copies of a random variable $\xi$ satisfying $\E\xi=0$, $\E\xi^2=1$, $\E\xi^4=\mu\ge1$
and $\lV\xi\rV_{\psi_{2}}\le K$. 
Then it holds that for $0<\alpha<1$,
\begin{equation}\label{ball}	
\begin{aligned}     
     \mathcal{Q}_{\alpha}\lk\mT\cap \mathbb{S}_{F};\pmb{\xi}\pmb{\xi}^{\mathsf{T}}\rk&:=\inf_{\Z\in\mT\cap \mathbb{S}_{F}}\mathbb{P}\lk\lv \lg\Z,\pmb{\xi}\pmb{\xi}^{\mathsf{T}}\rg\rv\ge\alpha\rk \\
	      & \gtrsim \lk1-\alpha^2\rk^{2}\min\left\{  \frac{4}{K^{8}},\lk\frac{\mu}{1+K^4}\rk^2,1\right\}.
	      \end{aligned}
		\end{equation}
\end{theorem}
\begin{proof}
We begin by noting two essential facts derived from Lemma 9 and Lemma 8 in \cite{krahmer2020complex}. 
Utilizing these two facts, we will demonstrate that for any $\mu\ge1$, and in particular when $\mu=1$, the small ball function is bounded away from zero.\\
\textbf{Fact 1}: The first is that for any $\Z\in\mathcal{S}^{n}$ it holds that
\begin{equation}\label{fact1}
		\E\lv\pp^{\mathsf{T}}\Z\pp\rv^{2}= \lk\trace \Z\rk^{2}+\lk\E\xi^4-1\rk\sum_{i=1}^{n}\Z_{i,i}^{2}+2\sum_{i\neq j}^{n}\Z_{i,j}^{2}.
	\end{equation}
By a direct calculation,
\begin{equation*}	
	\begin{aligned}
	\E\lv \pp^{\mathsf{T}}\Z\pp\rv^{2}&=\E\lk\sum_{i,j}\xi_i \Z_{i,j}\xi_j\rk^{2}
	                                     =\E\sum_{i,\tilde{i},j,\tilde{j}}\xi_i \xi_{\tilde{i}} \xi_j \xi_{\tilde{j}}\Z_{i,j}\Z_{\tilde{i},\tilde{j}}\\
                                             &=\E\xi^{4}\cdot\sum_{i=1}^{n}\Z_{i,i}^{2}+2\sum_{i\neq j}\Z_{i,j}^{2}+
                                             \sum_{i\neq\tilde{i}}\Z_{i,i}\Z_{\tilde{i},\tilde{i}}\\
                                             &=\lk\trace \Z\rk^{2}+\lk\E\xi^4-1\rk\sum_{i=1}^{n}\Z_{i,i}^{2}+2\sum_{i\neq j}^{n}\Z_{i,j}^{2}.
	 \end{aligned}
	\end{equation*}	
\textbf{Fact 2}: The second is that for any $\Z\in\mathcal{S}^{n}$ it holds that
\begin{equation}\label{fact2}
		\E\lv\pp^{\mathsf{T}}\Z\pp\rv^{4}\lesssim  \lk\trace \Z\rk^{4} + K^8 \lV\Z\rV_{F}^{4}.
		\end{equation}
	By the Hanson-Wright inequality in Lemma \ref{hanson}, we have that
	\begin{equation*}
	\begin{aligned}
	&\E \lv  \pp^{\mathsf{T}}\Z\pp - \E\left[ \pp^{\mathsf{T}}\Z\pp \right]\rv^{4}\\
	&= \int_{0}^{\infty} 4t^{3} \ \mathbb{P} \lk \lv  \pp^{\mathsf{T}}\Z\pp - \E\lz\pp^{\mathsf{T}}\Z\pp\rz\rv >  t  \rk  dt \\
	&\le 8 \int_{0}^{\infty} t^{3} \  \exp \lk-c  \frac{t^{2}}{K^{4} \lV\Z \rV^2_{F}} \rk dt   +  8\int_{0}^{\infty} t^{3} \  \exp \lk-c  \frac{t}{K^{2} \lV \Z\rV_{op}} \rk dt   \\
	&= 8 K^{8}\lV\Z\rV_{F}^{4}\int_{0}^{\infty} x^{3} \  \exp \lk -cx^{2} \rk dx  + 8K^{8}\lV\Z\rV_{op}^{4} \int_{0}^{\infty} x^{3} \  \exp \lk -cx \rk dx   \\
	&\lesssim  K^{8}\lV\Z\rV_{F}^{4}.
	\end{aligned}
	\end{equation*}
		Then the triangle inequality yields that
	\begin{equation*}	
	\begin{aligned}
	\E\lv \pp^{\mathsf{T}}\Z\pp\rv^{4}
	 & \lesssim
	 \lv \E \pp^{\mathsf{T}}\Z\pp\rv^{4}+\E \lv  \pp^{\mathsf{T}}\Z\pp - \E\left[ \pp^{\mathsf{T}}\Z\pp \right]\rv^{4}\\
	 &=\lk \trace\Z \rk^{4}+\E \lv \pp^{\mathsf{T}}\Z\pp - \E\left[ \pp^{\mathsf{T}}\Z\pp \right]\rv^{4}\\
	 &\lesssim\lk\trace\Z \rk^{4}+K^{8}\lV\Z\rV_{F}^{4}.
	 \end{aligned}
	\end{equation*}	

		We now turn to proving (\ref{ball}).
		For $\Z\in\mT\cap \mathbb{S}_{F}$, let the diagonal element of $\Z$ be $z_0$. Since $\lV\Z\rV_{F}=1$, we have $0\le \lv z_0\rv\le\frac{1}{\sqrt{n}}$.
		 By (\ref{fact1}), we have that
		  \begin{equation}\label{bbbb}
	 \begin{aligned}
		\E\lv\pp^{\mathsf{T}}\Z\pp\rv^{2}&= \lk\trace \Z\rk^{2}+\lk\E\xi^4-1\rk\sum_{i=1}^{n}\Z_{i,i}^{2}+2\sum_{i\neq j}^{n}\Z_{i,j}^{2}\\
		                                    &=\lk n z_0\rk^2+\lk\mu-1\rk n z_{0}^2+2\lk1-nz_{0}^{2}\rk\\
		                                    &=\lz n^{2}+\lk\mu-3\rk n \rz z_{0}^{2}+2.
		\end{aligned}
		 \end{equation}
Then it can be checked that $\E\lv\pp^{\mathsf{T}}\Z\pp\rv^{2}\ge1=\lV\Z\rV_{F}^2$. Thus, by the Paley-Zygmund inequality,
  \begin{equation*}
  \begin{aligned}
		\mathbb{P}\lk\lv\pp^{\mathsf{T}}\Z\pp\rv^{2}
		\ge\alpha^2\lV\Z\rV^{2}_{F}\rk&\ge\mathbb{P}\lk\lv\pp^{\mathsf{T}}\Z\pp\rv^{2}
		\ge 		\alpha^2\E\lv\pp^{\mathsf{T}}\Z\pp\rv^{2}\rk\\
		&\ge		\lk1-\alpha^2\rk^{2}\frac{\lk\E\lv\pp^{\mathsf{T}}\Z\pp\rv^{2}\rk^{2}}{\E\lv\pp^{\mathsf{T}}\Z\pp\rv^{4}}.
		\end{aligned}
		 \end{equation*}
		 
		 By (\ref{bbbb}) and (\ref{fact2}), we have 
\begin{equation*}
	 \begin{aligned}
		f\lk z_0\rk:=\inf_{\substack{\Z\in\mT\cap \mathbb{S}_{F}\\Z_{1,1}=z_0}}\frac{\lk\E\lv\pp^{\mathsf{T}}\Z\pp\rv^{2}\rk^{2}}{\E\lv\pp^{\mathsf{T}}\Z\pp\rv^{4}}
		&\gtrsim \inf_{\substack{\Z\in\mT\cap \mathbb{S}_{F}\\ \Z_{1,1}=z_0}}
\frac{ \lz\lk\trace \Z\rk^{2}+\lk\E\xi^4-1\rk\sum_{i=1}^{n}\Z_{i,i}^{2}+2\sum_{i\neq j}^{n}\Z_{i,j}^{2}
	\rz^{2} }{ \lk\trace \Z\rk^{4} + K^8 \lV\Z\rV_{F}^{4}}\\
	&= \frac{\lz \lk n^{2}+\lk\mu-3\rk n \rk z_{0}^{2}+2\rz^{2}}{\lk n z_0\rk^4+K^8}\\
	&\ge \lz\frac{ \lk n^{2}+\lk\mu-3\rk n \rk z_{0}^{2}+2}{\lk n z_0\rk^2+K^4}\rz^2\\
	&:=g\lk \lv z_0\rv\rk= \lk h\lk z^{2}_{0}\rk\rk^2.
		\end{aligned}
		 \end{equation*}	
		 
	By a direct calculation,
	\begin{equation}\label{h}
	\frac{\partial h\lk z^{2}_{0}\rk}{\partial \lk z^{2}_{0}\rk}=\frac{\lk n^{2}+\lk\mu-3\rk n \rk\cdot K^4-2n^2}{ \lk n^2 z_{0}^{2}+K^4\rk^2}.
		 \end{equation}
		Since the numerator in (\ref{h}) is constant and the denominator is positive, we have that $h\lk z^{2}_{0}\rk$ is monotonic on the interval $z^{2}_{0}\in\lz 0,1/n\rz$.	
	Besides, we have that for all $\mu\ge1$ and $z^{2}_{0}\in\lz 0,1/n\rz$, the numerator of $h\lk z^{2}_{0}\rk$ satisfies $\lk n^{2}+\lk\mu-3\rk n \rk z_{0}^{2}+2\ge0$.
	 	 Thus, $g$ is also monotonic on the interval $\lz 0,\frac{1}{\sqrt{n}}\rz$. Then we have that
		  \begin{equation*}
	 \begin{aligned}
		 \mathcal{Q}_{\alpha}\lk\mT\cap \mathbb{S}_{F};\pmb{\xi}\pmb{\xi}^{\mathsf{T}}\rk
		 &\ge\lk1-\alpha^2\rk^{2}\min f\lk z_0\rk\\
	         & \gtrsim\lk1-\alpha^2\rk^{2}\min\left\{ g\lk 0\rk,g\lk\frac{1}{\sqrt{n}} \rk\right\}\\
	         &\ge\lk1-\alpha^2\rk^{2}\min\left\{  \frac{4}{K^{8}},\lk\frac{n+\mu-1}{n+K^4}\rk^{2}\right\}\\
	         &\gtrsim\lk1-\alpha^2\rk^{2}\min\left\{  \frac{4}{K^{8}},\lk\frac{\mu}{1+K^4}\rk^2,1\right\}.
		\end{aligned}
		 \end{equation*}
\end{proof}

\subsection{Structured Random Matrix}
We investigate the spectral norm of the structured random matrix $\sum_{k=1}^{m}\mT\lk \pp_{k}\pp_{k}^{\mathsf{T}}\rk$,  which constitutes our third ingredient.
Our proof relies on embedding Toeplitz matrices into circulant matrices, which may inspire characterizations of the spectral norms of other structured random matrices.
 We provide the following theorem.
 \begin{theorem}\label{toep}
Let $\pp=\lk\xi_{1},\ldots,\xi_{n}\rk^{\mathsf{T}}$ be a random vector whose entries are i.i.d. copies of a random variable $\xi$ satisfying $\E\xi=0$, $\E\xi^2=1$ and $\lV\xi\rV_{\psi_{2}}\le K$.
Assume that $\left\{\pp_k\right\}_{k=1}^{m}$ are i.i.d. copies of $\pp$,
then we have
 \begin{equation}	   
   \E\lV\sum_{k=1}^{m}\mT\lk \pp_{k}\pp_{k}^{\mathsf{T}}\rk-m\E\mT\lk \pp\pp^{\mathsf{T}}\rk\rV_{op}\le CK^2\lk\sqrt{m}\log n+\log^{3/2} n\rk,
	         \end{equation}
	         where $C>0$ is a numerical absolute constant.
 \end{theorem}
 \begin{proof}
We divide the proof into four steps.

\textbf{Step 1: Expectation Calculation.} It can be seen that
  \begin{equation*}
\E\mT\lk \pp\pp^{\mathsf{T}}\rk=\mT\lk \E\pp\pp^{\mathsf{T}}\rk=\pmb{I}_{n}\in\mT.
	         \end{equation*}
We then set the random matrix
\begin{equation*}
\begin{aligned}
\Z_0:=\sum_{k=1}^{m}\mT\lk \pp_{k}\pp_{k}^{\mathsf{T}}\rk-m\E\mT\lk \pp\pp^{\mathsf{T}}\rk 
= \left(\begin{array}{cccc}
 z_{0}& z_{1}&\cdots &z_{n-1}\\
 z_{1}& z_{0}&\cdots &z_{n-2}\\
\vdots &  \vdots &\ddots& \vdots\\
 z_{n-1}&z_{n-2}&\cdots&z_{0}\\ 
\end{array}\right),
\end{aligned}
\end{equation*}
 and let   
$\pmb{z}:= \lk z_0,z_1,\ldots,z_{n-1}\rk^{\mathsf{T}}$. 
Since each entry on a diagonal of the Toeplitz projection is given by the average of the corresponding diagonal, we have that
\begin{align*}
{z}_{\ell}=
\left\{\begin{aligned}
&\frac{1}{n}\sum_{k=1}^{m}\sum_{i=1}^{n}\lk\xi^{2}_{k,i}-1\rk,\quad \ell=0;\\
&\frac{1}{n-\ell}\sum_{k=1}^{m}\sum_{i=1}^{n-\ell}\xi_{k,i}\xi_{k,i+\ell},\quad \ell=1,2,\ldots,n-1.
\end{aligned}\right.
\end{align*}  
 Apparently, one has $\E{z}_{\ell}=0$ for all $0\le\ell\le n-1$. 
 
\textbf{Step 2: Toeplitz Matrix Embedding.} The harmonic structure of Toeplitz matrices motivates us to embed $\Z_0$ into a new circulant matrix $\pmb{C}_{\Z_0}\in\R^{\lk 2n-1\rk\times \lk2n-1\rk}$ such that
 {\small\begin{equation*}
 \pmb{C}_{\Z_0}=
 \left(\begin{array}{cccc|ccc}
 z_{0}& z_{1}&\cdots & z_{n-1}&z_{n-1}&\cdots&z_{1}\\
 z_{1}& z_{0}&\cdots &z_{n-2}&z_{n-1}&\cdots&z_{2}\\
\vdots &  \vdots &\ddots& \vdots&\vdots&\cdots&\vdots\\
 z_{n-1}&z_{n-2}&\cdots&z_{0}&z_{1}&\cdots&z_{n-1}\\ \hline
 z_{n-1}&z_{n-1}&\cdots&z_{1}&z_{0}&\cdots&z_{n-2}\\
  \vdots&\vdots&\vdots&\vdots &\vdots &\ddots&\vdots\\
  z_{1}&z_{2}&\cdots &z_{n-1}&z_{n-2}&\cdots &z_{0}\\
\end{array}\right).
 \end{equation*}}
 Note that $\Z_{0}$ is a submatrix of $\pmb{C}_{\Z_0}$, we only need to give an upper bound for $\E\lV\pmb{C}_{\Z_0}\rV_{op}$. 
 
 Let $\omega=e^{\frac{2\pi \mathrm{i}}{2n-1}}$ be the $\lk2n-1\rk$-th root of unity.   
 Then the eigenvalues of $\pmb{C}_{\Z_0}$ are given by
 \begin{equation*}
 \begin{aligned}
\lambda_{j}&=z_0+z_1 \omega^j+z_2\omega^{2j}+\cdots+z_{n-1}\omega^{\lk n-1\rk j}
                       +z_{n-1}\omega^{nj}+\cdots+z_1\omega^{\lk 2n-2\rk j}\\
                    &=z_0+z_1\lk\omega^j+\omega^{\lk 2n-2\rk j}\rk +\cdots+  z_{n-1}\lk\omega^{\lk n-1\rk j}+\omega^{nj}\rk\\
                    &=z_0+2\sum_{\ell=1}^{n-1}z_{\ell}\cos\lk \frac{2\pi j\ell}{2n-1}\rk,\quad \text{for all}\ j=0,1,\ldots,2n-2.
\end{aligned}
\end{equation*}
It can be seen that $\E\lambda_{j}=0$ and
 \begin{equation*}
\lV\Z_0\rV_{op}\le\lV\pmb{C}_{\Z_0}\rV_{op}=\max_{0\le j\le2n-2}\lv\lambda_j\rv.
\end{equation*}

  \textbf{Step 3: Eigenvalue Estimation.} 
  Note that $\lambda_{j}$ is a centered quadratic form of $\left\{\pp_1,\cdots,\pp_m\right\}$.
  Let the symmetric coefficient matrix $\pmb{H}^{j}$ be
   \begin{equation}
\pmb{H}_{\alpha,\beta}^{j}:=\frac{1}{n-\lv\ell\rv} \cos\lk \frac{2\pi j\lv\ell\rv}{2n-1}\rk,\quad 1\le\alpha,\beta\le n\ \text{and} \ \alpha-\beta=\ell.
\end{equation}
Let 
\begin{equation*}
\widetilde{\pp}^{\mathsf{T}}\widetilde{\pmb{H}}^{j}\widetilde{\pp}:=\left(\begin{array}{cccc}
\pp_{1}^{\mathsf{T}}& \pp_{2}^{\mathsf{T}}&\cdots& \pp_{m}^{\mathsf{T}}
\end{array}\right)
\left(\begin{array}{cccc}
\pmb{H}^{j}& &&\\
 & \pmb{H}^{j}&&\\
 &&\ddots& \\
 &&&\pmb{H}^{j}\\ 
\end{array}\right)
\left(\begin{array}{cccc}
 \pp_1\\
 \pp_2 \\
\vdots \\
 \pp_m\\ 
\end{array}\right).
\end{equation*}
 Thus, $\lambda_{j}$ can be written as
  \begin{equation*}
\lambda_{j}
= \widetilde{\pp}^{\mathsf{T}}\widetilde{\pmb{H}}^{j}\widetilde{\pp}-
\E\left[\widetilde{\pp}^{\mathsf{T}}\widetilde{\pmb{H}}^{j}\widetilde{\pp}\right].
\end{equation*}
Then, by the Hanson--Wright inequality in Lemma \ref{hanson}, we can get
 \begin{equation*}
\mathbb{P}\lk\lv\lambda_j\rv\ge t\rk\le 2  \exp \lk-c\min \left\{ \frac{t^{2}}{K^{4} \lV\widetilde{\pmb{H}}^{j}\rV^{2}_{F}}  , \frac{t}{K^{2} \lV \widetilde{\pmb{H}}^{j} \rV_{op}}  \right\} \rk.
\end{equation*}
  It can be seen that $\lV\widetilde{\pmb{H}}^{j}\rV_{op}=\lV\pmb{H}^{j}\rV_{op}$ and $\lV\widetilde{\pmb{H}}^{j}\rV_{F}=\sqrt{m}\lV\pmb{H}^{j}\rV_{F}$.
  It therefore suffices to estimate $\lV\pmb{H}^{j}\rV_{F}$, since $\lV\pmb{H}^{j}\rV_{op}\le\lV\pmb{H}^{j}\rV_{F}$.
  Since $\pmb{H}^{j}$ is a symmetric Toeplitz matrix, we have
   \begin{equation*}
   \begin{aligned}
\lV\pmb{H}^{j}\rV^{2}_{F}
&=\sum_{1\le\alpha,\beta\le n}\lv\pmb{H}_{\alpha,\beta}^{j}\rv^2
\le\sum_{1\le\alpha,\beta\le n; \alpha-\beta=\ell} \frac{1}{\lk n-\lv\ell\rv\rk^{2}}\\
&=2\sum_{\ell=1}^{n-1}\frac{1}{n-\lv\ell\rv}+\frac{1}{n}\le2\log n+2.
\end{aligned}
\end{equation*}
Therefore, 
    \begin{equation}
     \left\{
     \begin{array}{rl}
   \lV\widetilde{\pmb{H}}^{j}\rV_{op}&\le\sqrt{2\log n+2},\\
    \lV\widetilde{\pmb{H}}^{j}\rV_{F}&\le \sqrt{m\lk2\log n+2\rk}.
     \end{array}
           \right.
     \end{equation} 
     
\textbf{Step 4: Uniform Argument.} We define the event $\Omega_{j,t}:=\left\{\lv\lambda_j\rv<t\right\}$
and define $\Omega_{t}=\bigcap\limits_{0\le j\le2n-2} \Omega_{j,t}$.
Then
  \begin{equation*}
   \begin{aligned}
\mathbb{P}\lk\max_{0\le j\le2n-2}\lv\lambda_j\rv\ge t\rk
&=\mathbb{P}\lk\Omega_{t}^{c}\rk\le\sum_{0\le j\le 2n-2}\mathbb{P}\lk\Omega_{j,t}^{c}\rk\\
&\le \sum_{0\le j\le 2n-2}2  \exp \lk-c\min \left\{ \frac{t^{2}}{K^{4} \lV\widetilde{\pmb{H}}^{j}\rV^{2}_{F}}  , \frac{t}{K^{2} \lV \widetilde{\pmb{H}}^{j} \rV_{op}}  \right\} \rk\\
&\le 4n\cdot  \exp \lk-c\min \left\{ \frac{t^{2}}{K^{4} m \lk2\log n+2\rk}  , \frac{t}{K^{2} \sqrt{2\log n+2}}  \right\} \rk.   
\end{aligned}
\end{equation*}
Finally, we have that 
  \begin{equation}
  \begin{aligned}
\E\lV\Z_0\rV_{op}&\le\E\lV\pmb{C}_{\Z_0}\rV_{op}=\E\max_{0\le j\le2n-2}\lv\lambda_j\rv\\
&\le\int_{0}^{\infty}\mathbb{P}\lk\max_{0\le j\le2n-2}\lv\lambda_j\rv\ge t\rk dt   \\
&\le \int^{C_1K^2\sqrt{m}\log n}_{0}1dt\ +\int_{C_1K^2\sqrt{m}\log n}^{\infty} 4n\cdot \exp \lk-c \frac{t^{2}}{2K^{4} m \log n}\rk dt   \\
&+\int^{C_2K^2\log^{3/2} n}_{0}1dt+\int_{C_2K^2\log^{3/2} n}^{\infty} 4n\cdot \exp \lk-c \frac{t}{K^{2} \sqrt{2\log n}} \rk dt \\
&\le \widetilde{C}_1K^2\sqrt{m}\log n+\widetilde{C}_2K^2\log^{3/2} n.\\
\end{aligned}
\end{equation}
 \end{proof}
 
  \subsection{Proof of Theorem \ref{thm}}
We now put all things together. 
By Theorem \ref{small}, the small ball function satisfies
    \begin{equation*}	     
    \begin{aligned}
     \mathcal{Q}_{2\alpha}&\lk \mK^{r}_{\mT}\cap \mathbb{S}_{F};\pmb{\xi}\pmb{\xi}^{\mathsf{T}}\rk
     \ge \mathcal{Q}_{2\alpha}\lk \mT\cap \mathbb{S}_{F};\pmb{\xi}\pmb{\xi}^{\mathsf{T}}\rk\\
      &\gtrsim \lk1-4\alpha^2\rk^{2}\min\left\{  \frac{4}{K^{8}},\lk\frac{\mu}{1+K^4}\rk^2,1\right\}.
     \end{aligned}
	  \end{equation*}	
	Moreover, for the empirical process term, we have
		      \begin{equation*}
		      \begin{aligned}
		       \E&\lV\sum_{k=1}^{m}\varepsilon_{k}\mT\lk\pp_k\pp_k^{\mathsf{T}}\rk \rV_{op}\\
		       &\le   \E_{\varepsilon} \E_{\pp}\lV\sum_{k=1}^{m}\varepsilon_{k}\lz\mT\lk\pp_k\pp_{k}^{\mathsf{T}}\rk -\E_{\pp}\mT\lk\pp_k\pp_k^{\mathsf{T}}
		       \rk\rz\rV_{op}+ 
		       \E_{\varepsilon}\E_{\pp}\lV\sum_{k=1}^{m}\varepsilon_{k}\E_{\pp}\mT\lk\pp_{k}\pp_{k}^{\mathsf{T}}\rk \rV_{op}\\
		       &\le2  \E_{\pp}\lV\sum_{k=1}^{m}\lz\mT\lk \pp_{k}\pp_{k}^{\mathsf{T}}\rk-\E\mT\lk \pp_k\pp_{k}^{\mathsf{T}}\rk\rz\rV_{op}
		       + \E_{\varepsilon}\lV\sum_{k=1}^{m}\varepsilon_{k}\pmb{I}_n \rV_{op}\\
		       &\le 2CK^2\lk\sqrt{m}\log n +\log^{3/2} n\rk +\widetilde{C}\sqrt{m}\\
		       	    &\lesssim \lk K^2+1\rk\lk\sqrt{m}\log n +\log^{3/2} n\rk.        
		        \end{aligned}
		           \end{equation*} 
		In the second inequality, we have used the symmetrization principle \cite[Lemma 6.4.2]{vershynin2018high} and $\E_{\pp}\mT\lk\pp_k\pp_k^{\mathsf{T}}\rk=\pmb{I}_n$.
		            In the third inequality, we have used Theorem \ref{toep} and the estimate $\E_{\varepsilon}\lV\sum_{k=1}^{m}\varepsilon_{k}\pmb{I}_n \rV_{op}=\E_{\varepsilon}\lv \sum_{k=1}^{m}\varepsilon_{k} \rv\lesssim\sqrt{m}$.
Thus, we can get 
  \begin{equation}
   \E\lV\frac{1}{m}\sum_{k=1}^{m}\varepsilon_{k}\mT\lk\pp_k\pp_k^{\mathsf{T}}\rk \rV_{op}		
   \lesssim \lk K^2+1\rk \lk\frac{\log n}{\sqrt{m}}+\frac{\log^{3/2} n}{m}\rk. 
   \end{equation}          	           
Now, by Proposition \ref{des}, the supremum of the empirical process satisfies
\begin{equation*}	    
  \begin{aligned}
	            \mathcal{R}_{m}\lk \mK^{r}_{\mT}\cap \mathbb{S}_{F};\pp\pp^{\mathsf{T}}\rk
	            &=\E\sup_{\Z\in\mK^{r}_{\mT}\cap \mathbb{S}_{F}}\lv\frac{1}{m}\sum_{k=1}^{m}\varepsilon_{k}\left\langle\Z,\pp_k\pp_k^{\mathsf{T}}\right\rangle\rv\\
	            &=\E\sup_{\Z\in\mK^{r}_{\mT}\cap \mathbb{S}_{F}}\lv\left\langle\Z,\frac{1}{m}\sum_{k=1}^{m}\varepsilon_{k}\mT\lk\pp_k\pp_k^{\mathsf{T}}\rk\right\rangle\rv\\
	            &\le\lk\sqrt{2}+1\rk \sqrt{r}\cdot \E\lV\frac{1}{m}\sum_{k=1}^{m}\varepsilon_{k}\mT\lk\pp_k\pp_k^{\mathsf{T}}\rk \rV_{op}\\
	            &\lesssim \lk K^2+1\rk  \sqrt{r}\cdot \lk\frac{\log n}{\sqrt{m}}+\frac{\log^{3/2} n}{m}\rk.
	            \end{aligned}
	      	         \end{equation*} 
		      
Finally, we choose $\alpha=1/4$ and $ t=c_2\sqrt{m}$ in Lemma \ref{small ball}. 
By (\ref{small1}), provided $m\ge L r\log^2 n$ for a sufficiently large constant $L$, we have 
\begin{equation*}
\lambda_{\min}\lk\mA; \mK^{r}_{\mT}\cap \mathbb{S}_{F}\rk\gtrsim \tilde{c}m^{1/p}
 \end{equation*} 
 with probability exceeding $1-e^{-cm}$.
Here $c_2,L,c$ and $\tilde{c}$ are positive constants depending only on $K$ and $\mu$. 
Thus, by (\ref{aaa}) we have finished the proof.

\normalem
\bibliographystyle{plain}
\bibliography{ref}

\end{document}